\documentclass[pdflatex,sn-mathphys-num]{sn-jnl}

\usepackage{graphicx}%
\usepackage{multirow}%
\usepackage{amsmath,amssymb,amsfonts}%
\usepackage{amsthm}%
\usepackage{mathrsfs}%
\usepackage[title]{appendix}%
\usepackage{xcolor}%
\usepackage{textcomp}%
\usepackage{manyfoot}%
\usepackage{booktabs}%
\usepackage{algorithm}%
\usepackage{algorithmicx}%
\usepackage{algpseudocode}%
\usepackage{listings}%

\usepackage{hyperref}          
\usepackage[nameinlink,capitalize]{cleveref}  

\theoremstyle{thmstyleone}%
\newtheorem{theorem}{Theorem}[section]
\newtheorem{proposition}{Proposition}[section] 
\newtheorem{corollary}{Corollary}[section]
\newtheorem{lemma}{Lemma}[section] 

\theoremstyle{thmstyletwo}
\newtheorem{example}{Example}[section]
\newtheorem{remark}{Remark}[section]

\theoremstyle{thmstylethree}
\newtheorem{definition}{Definition}[section]  
  
\crefname{theorem}{Theorem}{Theorems}
\crefname{lemma}{Lemma}{Lemmas} 
\crefname{proposition}{Proposition}{Propositions}
\crefname{corollary}{Corollary}{Corollaries}
\crefname{definition}{Definition}{Definitions}
\crefname{remark}{Remark}{Remarks}
\crefname{example}{Example}{Examples}

\newcommand{\va}{\boldsymbol{\alpha}}
\newcommand{\vc}{\boldsymbol{c}}
\newcommand{\ve}{\boldsymbol{e}}
\newcommand{\vu}{\boldsymbol{u}}
\newcommand{\vv}{\boldsymbol{v}}
\newcommand{\vx}{\boldsymbol{x}}
\newcommand{\vy}{\boldsymbol{y}}
\newcommand{\T}{\mathrm{T}}

\newcommand{\rk}{\mathrm{rank}}

\newcommand{\f}{\mathbb{F}}
\newcommand{\fq}{\mathbb{F}_{q}}
\newcommand{\fqn}{\mathbb{F}_{q}^{n}}
\newcommand{\fqx}{\mathbb{F}_{q}[x]}
\newcommand{\fqstar}{\mathbb{F}_{q}^{*}}

\usepackage{enumerate}

\usepackage{makecell} 

\allowdisplaybreaks
\begin{document}

\title{Non-GRS type MDS and AMDS codes from extended TGRS codes}


\author[1]{\fnm{Meiying} \sur{Zhang}}\email{zmy0525zmy@163.com}

\author*[1]{\fnm{Shudi} \sur{Yang}}\email{yangshudi@qfnu.edu.cn}

\author[1]{\fnm{Yanbin} \sur{Zheng}}\email{zheng@qfnu.edu.cn}

\affil*[1]{\orgname{School of Mathematical Sciences, Qufu Normal University}, \orgaddress{\city{Shandong}, \postcode{273165}, \country{China}}}


\abstract{Maximum distance separable (MDS) and almost maximum distance separable (AMDS) codes have been widely used in various fields such as communication systems, data storage, and quantum codes because of their algebraic properties and excellent error-correcting capabilities. In this paper, we construct a class of extended twisted  generalized Reed-Solomon (TGRS) codes and determine the necessary and sufficient conditions for these codes to be MDS or AMDS. Additionally, we prove that these codes are not equivalent to generalized Reed-Solomon (GRS) codes. As an application, under certain circumstances, we compute the covering radii and deep holes of these codes.}




\maketitle

\section{Introduction}\label{sec1}

Let $ \fq $ be a finite field with $ q $ elements, where $ q $ is a prime power. Denote $\fqstar = \fq \backslash \{0\}$, and $\fqx$ the polynomial ring over $ \fq $. 
An $ [n,k,d] $ linear code $ C $ over $ \fq $ is a $ k $-dimensional subspace of $ \fqn $ with length $ n $ and minimum distance $ d $.
The dual code of $ C $, denoted by $ C^\perp $, consists of all vectors $ \vx \in \fqn $ such that $ \vx \cdot \vc^\T = 0 $ for all $ \vc \in C $.
For an $ [n,k,d] $ linear code $ C $, if $ d = n - k + 1 $, i.e., the Singleton bound is attained, then the code $ C $ is named a maximum distance separable (in short, MDS) code; 
if $ d = n - k $, then the code $ C $ is referred to as almost MDS (in short, AMDS).
The code $ C $ is called near MDS (in short, NMDS) if both $ C $ and $ C^{\perp} $ are AMDS.

Since MDS codes and AMDS codes are of great importance in coding theory and applications \cite{ CadambeHL2011, DingT2020, FECC2003, TTECC1977, SakakibaraT2013, SimosV2018, XuCQ2022},
the study of these codes has attracted significant attention \cite{Boer1996, Ballico2001, DodunekovaDK1997, Faldum1997, GengYZZ2022, KaiZL2015}. 
It is well-known that Reed-Solomon (in short, RS) codes and generalized Reed-Solomon (in short, GRS) codes are two classes of MDS codes, but twisted generalized Reed-Solomon (in short, TGRS) codes are not necessarily MDS.
Therefore, in recent years, many researchers have focused on constructing MDS or AMDS codes based on TGRS codes or extended TGRS (in short, ETGRS) codes \cite{GuZ2023, HuangYNX212195, LiZS2025, SuiYLH2022, SuiZS221937, YangWW2025, ZhuL24102395, ZhuL2024, ZhangZT221649}.

Furthermore, non-GRS type MDS codes and AMDS codes have also garnered substantial due to their ability to resist Wieschebrink and Sidelnikov-Shestakov attacks \cite{BeelenBPR2018, LavauzelleR2020}.
Roth and Lempel \cite{RL1989} proposed the first construction of non-GRS MDS codes by adding two columns to the generator matrices of GRS codes.
Han et~al.\ \cite{HanF2023} derived a necessary and sufficient condition for Roth-Lempel codes and its dual codes to be AMDS. 
In addition, Beelen et~al.\ \cite{BeelenPN2017} first introduced the concept of TGRS codes in 2017, which can be viewed as a generalization of GRS codes. 
They established a sufficient and necessary condition for the TGRS codes to be MDS and showed that most of these TGRS codes are non-GRS type.
Based on this work, the properties of TGRS codes, including MDS, AMDS, non-GRS type, self-dual, and self-orthogonal, have been extensively investigated in \cite{HuangYN23221,HuangYNX212195,  LiuL212051,SuiZS221937}. 
Following this research direction, Zhu et~al.\ \cite{ZhuL2024} focused on a class of ETGRS codes by adding a column to the generator matrices of TGRS codes. They provided a necessary and sufficient condition for ETGRS codes to be MDS or AMDS, 
determined their weight distribution based on the subset sum problem, and proved that these codes are neither GRS nor EGRS codes.
By adding a column, which is different from the one used in \cite{ZhuL2024}, to the generator matrix of the TGRS code, Zhang et~al.\ \cite{ZhangD2026} constructed another class of ETGRS codes and developed several classes of self-orthogonal MDS codes, 
which are also non-GRS type. 
Subsequently, by adding different columns to the generator matrix of the same TGRS code, Li et~al.\ 
\cite{LiZS2025} constructed two types of ETGRS codes and then studied their non-GRS type MDS properties and further directly derived some results on the existence of error-correcting pairs of non-GRS type MDS codes.
There are two other fundamental and crucial objects in the study regarding non-GRS type MDS codes and AMDS codes, namely, covering radii and deep holes, as they have important applications in many aspects \cite{CohenKMS1985, ElimelechFS2021}.
For more details on this topic, the reader is referred to \cite{Bartoli2015,Dur1994,KetiW2016,LiZS2025,WuDC2025,ZhuangCL2016,ZhangWK} and the references therein.

Inspired by above works, this paper mainly focuses on a class of linear codes defined in \cref{def: C}.
The structure of this paper is as follows. 
\cref{sec2} reviews some basic notations and known results. 
\cref{sec3} proves that the constructed codes are non-GRS type. 
\cref{sec4} provides the necessary and sufficient conditions for the codes to be MDS or AMDS, along with concrete examples of non-GRS type MDS codes and AMDS codes.
As an application, we present the results of the covering radii and deep holes of our codes in \cref{sec5}. 
Finally, we conclude this paper in \cref{sec6}.
\section{Preliminaries}\label{sec2}

In this section, we review several definitions and known results related to linear codes.  
Now, we explain some notations used in this paper:
\begin{itemize}
\item For a positive integer $ n $, $ [n] $ denotes the set $ \{1, 2, \dots, n \} $.
\item For any set $ I \subseteq [n] $, $ |I| $ denotes the size of $ I $.
\item For any square matrix $ B $, denote by $ \det(B) $ the determinant of $ B $.
\item For any matrix $ A $, $ \rk(A) $ and $A^{\T}$ denote the rank and transpose of $A$, respectively.
\item For any $ \vx, \vy \in \fqn $, $ wt(\vx) $ denotes the (Hamming) weight of $ \vx $; $ d(\vx, \vy) $ denotes the (Hamming) distance betweem $\vx$ and $\vy$.
$ \dim(C) $ and $ d(C) $ denote the dimension and the minimum distance of linear code $ C $, respectively.
\item $ \vv = \boldsymbol{1} $ denotes the all-ones vector $ (1, 1, \dots, 1) $ with length $ n $ or $ n+2 $ depending on the context.
\item For any  $ a \in \fq $ and $ \va = (\alpha_1, \dots, \alpha_n) \in \fqn $, denote $ (\va, a)  = (\alpha_1, \dots, \alpha_n, a) $.
\item The $ \fq $-linear subspace generated by the subset $ S $ of $ \fqn $ is denoted by $ \langle S \rangle $.
\end{itemize}

\subsection{Basic definitions of linear codes}

First, we recall the definition of GRS codes.
\begin{definition}[{\cite[Page 176]{FECC2003}}]\label{def:GRS}
Let $ k $ and $ n $ be integers with $ 1 \leqslant k \leqslant n \leqslant q $, $ \va =  (\alpha_1 $, $ \alpha_2 $, \dots, $ \alpha_n) \in \fqn $ with $ \alpha_i \ne \alpha_j (i \ne j)$, and $ v_1 $, $ v_2 $, \dots, $ v_n \in \fqstar $. Then the GRS code with length $ n $ and dimension $ k $ is defined as 
\[
\mathcal{GRS}_{k,n}(\va,\vv) = \{(v_1 f(\alpha_1), \dots, v_n f(\alpha_n)) \colon f(x) \in \fqx, \deg f(x) \leqslant k-1 \}.
\]
If $ \vv = \boldsymbol{1} $, the GRS code reduces to an RS code.
\end{definition}

Next, we define the linear space $ \mathcal{V}_{k, t, h, \eta} $ of the twisted polynomials.
\begin{definition}[{\cite[Definition 5]{BeelenPN2017}}]
Let $ k, t $, $h $ be integers with $ 0 \leqslant h<k \leqslant q $ and let $ \eta \in \fqstar $. Then the set of $ (k, t, h, \eta) $-twisted polynomials is defined as 
\[
\mathcal{V}_{k, t, h, \eta} = \{ f(x) = \sum_{i=0}^{k-1} a_i x^i + \eta a_h x^{k-1+t} \colon a_i \in \fq (i = 0, \dots, k-1)\},
\]
where $ h $ and $ t $ are the hook and twist, respectively.
\end{definition}

Based on the linear space $ \mathcal{V}_{k, t, h, \eta} $ of the twisted polynomials, we define TGRS codes as follows.
\begin{definition}[{\cite[Definition 6]{BeelenPN2017}}]
For any integers $ k, t, h, n $ with $0 \leqslant h \leqslant k-1 < k - 1 + t < n \leqslant q $, let $ \eta \in \fqstar $, $ \va = (\alpha_1, \dots, \alpha_n) \in \fqn$ with $ \alpha_i \neq \alpha_j (i \neq j) $ and $ \vv = (v_1, \dots, v_n) \in (\fqstar)^n $. 
Then the TGRS code with length $ n $ and dimension $ k $ is defined as 
\[
\mathcal{C}_{k,n,t,h}(\va, \vv, \eta) = \{(v_1 f(\alpha_1), \dots, v_n f(\alpha_n)) \colon f(x) \in \mathcal{V}_{k, t, h, \eta} \}.
\]
\end{definition}

\subsection{The Schur product}

This subsection reviews the definition of the Schur product and related results.
\begin{definition}
Let $ \vx = (x_1, \dots, x_n) $, $ \vy = (y_1, \dots, y_n) \in \fqn$. Then the Schur product of $ \vx $ and $ \vy $ is defined as 
\[
  \vx * \vy = (x_1 y_1, \dots, x_n y_n).
\]
The Schur product of two $ q $-ary codes $ C_1 $ and $ C_2 $ with length $ n $ is defined as
\[
C_1 * C_2 = \langle \vc_1 * \vc_2 : \vc_1 \in C_1, \vc_2 \in C_2 \rangle.
\]
In particular, for an $ [n,k] $ linear  code $ C $, the Schur product of  $C$ with itself is denoted by $ C^2 $.
\end{definition}

\begin{lemma}[{\cite[Remark 2.2]{ZhuL2024}}]\label{lem: Schurprod_C}
For any linear codes $ C_1 $ and $ C_2 $, if $ C_1 = \langle \boldsymbol{\beta}_1, \dots, \boldsymbol{\beta}_{k_1} \rangle $ and $ C_2 = \langle \boldsymbol{\gamma}_1, \dots, \boldsymbol{\gamma}_{k_2} \rangle $ with $ \boldsymbol{\beta}_i$, $ \boldsymbol{\gamma}_j \in \fqn $ for $  i = 1, \dots, k_1$ and $ j = 1, \dots, k_2 $, then 
\[
C_1 * C_2 = \langle \boldsymbol{\beta}_i * \boldsymbol{\gamma}_j :i = 1, \dots, k_1, j = 1, \dots, k_2 \rangle . 
\]
\end{lemma}

The following lemma describes the dual of a GRS code.
\begin{lemma}[{\cite[Lemma 2.3(i)]{JinX17}}]\label{lem: GRS1}
Let $ \vu = (u_1, \dots, u_n) $ with $ u_i = \prod\limits_{j = 1, j \ne i}^{n} (\alpha_i - \alpha_j)^{-1} $ for $i=1,\ldots,n $. Then
\[
\mathcal{GRS}_{k, n}^{\perp} (\va, \boldsymbol{1}) = \mathcal{GRS}_{n-k, n} (\va, \vu).
\]
\end{lemma}

By Definitions \cref{def:GRS} and \cref{lem: Schurprod_C}, the Schur square of a GRS code is also a GRS code, as stated below. 

\begin{lemma}[{\cite[Proposition 10]{MarquezMP13}}]\label{lem: GRS2}
Let $ \vu = (u_1, \dots, u_n) $ with $ u_i = \prod\limits_{j = 1, j \ne i}^{n} (\alpha_i - \alpha_j)^{-1} $ for $i=1,\ldots,n $. 
If $ 2 \leqslant k \leqslant \frac{n+1}{2} $, then
\[
\mathcal{GRS}_{k, n}^{2} (\va, \vu) = \mathcal{GRS}_{2k-1, n} (\va, \vu^2),
\]
where $\vu^2= (u_1^2, \dots, u_n^2)$ for simplicity.
\end{lemma}

Combining Lemmas \cref{lem: GRS1} and \cref{lem: GRS2}, we obtain the following result.
\begin{proposition}\label{prop: GRS^2}
Let $ \vu = (u_1, \dots, u_n) $ with $ u_i = \prod\limits_{j = 1, j \ne i}^{n} (\alpha_i - \alpha_j)^{-1} $ for $i=1,\ldots,n $. 
If $ \frac{n-1}{2} \leqslant k \leqslant n-2 $, then
\[
(\mathcal{GRS}_{k, n}^{\perp} (\va, \boldsymbol{1}))^2 = \mathcal{GRS}_{2n-2k-1, n} (\va, \vu^2).
\]
\end{proposition}

\subsection{A kind of extended linear codes}

In \cite{SunDC2024}, Sun et~al.\ provided a general description of extended codes. 
Let $\ve = (e_1, e_2, \dots, e_n) \in \fqn$ be any nonzero vector. 
Any given $ [n, k, d] $ code $ C $ over $ \fq $ can be extended to an $ [n+1, k, \bar{d}] $ code $ \overline{C} (\ve) $ over $ \fq $ as follows: 
\begin{align}\label{eq: bar_C}
\overline{C} (\ve) = \{(c_1, \dots, c_n, c_{n+1}) \colon (c_1, \dots, c_n) \in C, c_{n+1} = \sum_{i = 1}^{n} e_i c_i \},
\end{align}
where $ \bar{d} = d $ or $ \bar{d} = d + 1 $.

All extended codes discussed in this paper follow the definition in \eqref{eq: bar_C}.

The following lemma describes the relationship between the generator (and parity-check) matrices of a code and its extended code.

\begin{lemma}[{\cite[Page 13]{SunDC2024}}]\label{lem: bar_C}
Let $ C $ be an $ [n, k, d] $ linear code over $ \fq $ and let $ \ve \in \fqn$. 
If $ C $ has a generator matrix $ G $ and a parity-check matrix $ H $, then the generator matrix and parity-check matrix for the extended code $ \overline{C} (\ve) $ defined in \eqref{eq: bar_C} are $ (G, G \ve^\T) $ and
\[
  \begin{pmatrix}
    H & 0^\T  \\
    \ve & -1
  \end{pmatrix},
\]
respectively, where $ \ve^\T$ denotes the transpose of $ \ve $. 
\end{lemma}

Recall that a monomial matrix is a square matrix with exactly one nonzero entry in each row and each column {\cite[Section 1.7]{FECC2003}}.

\begin{definition}\label{def: Mon-Equ}
Let $ P $ and $ Q $ be two linear codes of the same length over $ \fq $, 
and let $ G $ be a generator matrix of $ P $. 
Then $ P $ and $ Q $ are monomially equivalent if there is a monomial matrix $ M $ 
such that $ GM  $ is a generator matrix of $ Q $. 
\end{definition}

The following lemma establishes the relationship between the non-GRS type property of an extended code and its original code.

\begin{lemma}\label{lem: barC_GRS_C}
If the extended code $ \overline{C}(\ve) $ is monomially equivalent to a GRS code, 
then the original code $ C $ is also monomially equivalent to a GRS code.
\end{lemma}

Throughout this paper, a code is called non-GRS type if it is not monomially equivalent to a GRS code.

\subsection{Covering radius and deep holes}

Recall that a sphere of radius $t$ in $\fqn$ is a set of vectors in $\fqn$ at distance no more than $t$ from a given vector in $\fqn$. 
An $ [n, k, d] $ code $ C $ over $ \fqn $ can correct $ t = \lfloor (d-1)/2 \rfloor $ errors. 
Thus spheres in $ \fqn $ of radius $t$ centered at codewords are pairwise disjoint, but this does not hold for spheres of larger radius.
It is natural to explore the opposite situation of finding the smallest radius, called the covering radius, of spheres centered at codewords that completely cover the space $ \fqn $. 

In general, given a code, it is difficult to find its covering radius, and the complexity of this has been investigated; see \cite{CohenHLL97}.
In this subsection, we will present some basic theories and properties of the covering radius. 
Firstly, we give the following definition of the covering radius.

\begin{definition}\label{def: coveringradius}
The covering radius of a code $ C $, denoted by $ \rho(C) $, 
is defined to be the maximum distance from any vector in $\fqn$ to the nearest codeword in $ C $. 
Equivalent, 
\[
\rho(C) = \max\limits_{\vx \in \fqn} \min\limits_{\vc \in C} d(\vx, \vc).
\]
A deep hole of $C$ is a vector achieving this covering radius.
\end{definition}

\begin{lemma}[Redundancy Bound {\cite[Corollary 11.1.3]{FECC2003}}]\label{lem: RedundancyBound}
Let $ C $ be an $ [n, k] $ code. Then $ \rho(C) \leqslant n - k $.
\end{lemma}

\begin{lemma}[Supercode Lemma {\cite[Lemma11.1.5]{FECC2003}}]\label{lem: SupercodeLemma}
If $ C $ and $ C' $ are linear codes with $ C \subseteq C' $, then
$ \rho(C) \geqslant \min\{wt(\vc) \colon \vc \in C' \backslash C \} $.    
\end{lemma}

The following lemmas characterize the covering radius and deep holes of a linear code using its generator matrix.
\begin{lemma}[{\cite[Lemma II.7]{ZhangWK}}]\label{lem: coveringradius} 
Let $ G $ denote a generator matrix of an $ [n,k] $ MDS code $ C $. 
Then the covering radius $ \rho(C) = n-k $ if and only if there exists a vector $ \vx \in \fqn $ such that the $ (k+1) \times n $ matrix $ \big( \frac{G}{\vx} \big) $ generates an $ [n,k+1] $ MDS code.
\end{lemma}

\begin{lemma}[{\cite[Lemma 6]{WuDC2025}}]\label{lem: deephole2} 
Suppose $ G $ is a generator matrix of an $ [n,k] $ MDS code $ C $ over $ \fq $ with covering radius $ \rho(C) = n-k $. 
Then a vector $ \vu \in \fqn $ is a deep hole of $ C $ if and only if $ G' = \big( \frac{G}{\vu} \big) $ generates an MDS code.
\end{lemma}

The following lemma establishes a relationship between extended codes and deep holes of MDS codes.
\begin{lemma}[{\cite[Theorem 1]{WuDC2025}}]\label{lem: deephole} 
Let $ C $ be an $[n, k]$ MDS code over $\fq$. 
Then for any $ \ve \in \fqn $, the extended code $ \overline{C}(\ve) $ in \eqref{eq: bar_C} is MDS if and only if $ \rho(C^{\perp}) = k $ and $ \ve $ is a deep hole of the dual code $ C^{\perp} $.
\end{lemma}

\subsection{Some auxiliary results}

Generally speaking, a linear code is not necessarily MDS or AMDS. 
However, the following lemmas provide necessary and sufficient conditions for a linear code to be  MDS or AMDS, respectively. 

\begin{lemma}[{\cite[Theorem~2.4.3]{FECC2003}}]\label{lem: MDS}
Let $ G $ be the generator matrix of an $ [n,k,d] $ linear code $ C $, 
then $ C $ is MDS if and only if any $ k $ columns of G are linearly independent.
\end{lemma}

\begin{lemma}[{\cite[Lemma~3.7]{SuiYLH2022}}]\label{lem: AMDS}
Let $ G $ be the generator matrix of an $ [n,k,d] $ linear code $ C $, 
then $ C $ is AMDS if and only if there exist $ k $ columns in $ G $ whose rank is at most $ k - 1 $ 
and the rank of any $ k + 1 $ columns of $ G $ is exactly $ k $.
\end{lemma}

The following result will be frequently used in the subsequent sections.
\begin{lemma}\label{lem: (Van-M)u=1}
Let $ u_i = \prod\limits_{j = 1, j \ne i}^{n} (\alpha_i - \alpha_j)^{-1} $ for $ i = 1, \ldots, n $. 
For any subset $ \{\alpha_1, \dots, \alpha_n \} \subseteq \fq $ with $ n \geqslant 2 $, it holds that 
\begin{align*}
\sum_{i=1}^{n} u_i \alpha_i^{\ell} = 
\begin{cases}
0, & 0 \leqslant \ell \leqslant n-2, \\
1, & \ell = n-1, \\
\sum\limits_{i=1}^{n} \alpha_i, & \ell = n.
\end{cases}
\end{align*}
\end{lemma}


\section{A family of non-GRS type linear codes}\label{sec3}

This section defines a family of linear codes based on the extended TGRS code and proves that these codes are non-GRS type.


\begin{definition}\label{def: C}
Let $ n $, $ k $ and $ h $ be positive integers satisfying $ 0 \leqslant h \leqslant k-2 $ and $ 3 \leqslant k < k+1 < n \leqslant q $. 
Let $ \eta $, $ \delta \in \fqstar $, $ \va =  (\alpha_1 $, $ \alpha_2 $, \dots, $ \alpha_n) \in \fqn $ with $ \alpha_i \ne \alpha_j $ for $ i \ne j $, and $ v_1 $, $ v_2 $, \dots, $ v_n \in \fqstar $. 
Denote by $ \mathcal{C} $ the extended TGRS code over $ \fq $ generated by the following matrix:
\begin{align}\label{matrix: G}
	G = \begin{pmatrix}
	v_1 & \dots & v_n & 0 & 0 \\
	v_1 \alpha_1  & \dots & v_n \alpha_n & 0 & 0 \\
	\vdots   &  &  \vdots &  \vdots &  \vdots \\
	v_1 \alpha_1^{h-1}  & \dots & v_n \alpha_n^{h-1} & 0 & 0 \\
	v_1(\alpha_1^h + \eta \alpha_1^{k+1}) & \dots & v_n(\alpha_n^h + \eta \alpha_n^{k+1}) & 1 & 1 \\
	v_1 \alpha_1^{h+1}  & \dots & v_n \alpha_n^{h+1} & 0 & 0 \\
	\vdots   &  &  \vdots &  \vdots &  \vdots \\
	v_1 \alpha_1^{k-2}  & \dots & v_n \alpha_n^{k-2} & 0 & 0 \\
	v_1 \alpha_1^{k-1} & \dots & v_n \alpha_n^{k-1} & 0 & \delta
	\end{pmatrix}.
\end{align} 
Let 
\[
V = \big\{f(x) = \sum_{i=0}^{k-1} a_i x^i + \eta a_h x^{k+1} \colon a_i \in \fq \text{ for } i = 0, 1, \dots, k-1 \big\}.
\]
Then the extended TGRS code $ \mathcal{C} $ can be represented as 
\[
\mathcal{C} = \big\{(v_1 f(\alpha_1), \dots, v_n f(\alpha_n), a_h, a_h + \delta a_{k-1}) \colon f(x) \in V \big\},
\]
where $ a_h $ and $ a_{k-1} $ are the coefficients of $ x^h $ and $ x^{k-1} $ in $ f(x) $, respectively.
\end{definition}

From the above definition, $ \mathcal{C} $ has parameters $ [n + 2, k] $.
In order to better study the non-GRS type property of $ \mathcal{C} $, 
we introduce a punctured version of $ \mathcal{C} $, denoted by $ C_1 $,
which has parameters $ [n + 1, k] $ and is generated by the following matrix:
\[
G_1 =
\begin{pmatrix}\label{matrix: G_1}
  v_1 & \dots & v_n & 0 \\
  v_1 \alpha_1 & \dots & v_n \alpha_n & 0 \\
  \vdots &  & \vdots & \vdots \\ 
  v_1 \alpha_1^{h-1} & \dots & v_n \alpha_n^{h-1} & 0 \\
  v_1 (\alpha_1^h + \eta \alpha_1^{k+1}) & \dots & v_n (\alpha_n^h + \eta \alpha_n^{k+1}) & 1 \\
  v_1 \alpha_1^{h+1} & \dots & v_n \alpha_n^{h+1} & 0 \\
  \vdots &  & \vdots & \vdots \\ 
  v_1 \alpha_1^{k-1} & \dots & v_n \alpha_n^{k-1} & 0 \\
\end{pmatrix},
\]
where $ n $, $ k $ and $ h $ are positive integers satisfying $ 0 \leqslant h \leqslant k-2 $ and $ 3 \leqslant k < k+1 < n \leqslant q $, $ \eta $, $ \delta \in \fqstar $, $ \alpha_1 $, $ \alpha_2 $, \dots, $ \alpha_n \in \fq $ are pairwise distinct, and $ v_1 $, $ v_2 $, \dots, $ v_n \in \fqstar $.

The code $ \mathcal{C} $ is an extended code of $ C_1 $, as shown in the following theorem.

\begin{theorem}\label{thm: C=barC_1}
The code $ \mathcal{C} = \overline{C_1}(\ve) $, 
where $ \ve = (e_1, \dots, e_n, e_{n+1}) \in \fq^{n+1} $ with
\[
   e_i = \frac{\delta  \alpha_i^{n-k}}{ v_i  \prod\limits_{j = 1, j \ne i}^{n} (\alpha_i - \alpha_j)}  \quad \text{for all} \quad 1 \leqslant i \leqslant n
\]  
and 
\[
 e_{n+1} = 1 + \delta \eta \big(\sum\limits_{1 \leqslant i < j \leqslant n} \alpha_i \alpha_j - (\sum_{i = 1}^{n} \alpha_i)^2 \big).
\]   
\end{theorem}
\begin{proof}
This follows easily from \cref{lem: bar_C}.
\end{proof}

Before proving that $ \mathcal{C} $ is non-GRS type, we first prove that $ C_1 $ is non-GRS type.
\begin{theorem}\label{thm: C_1noGRS}
For $ 3 \leqslant k < \frac{n+2}{2}$, the code $ C_1 $ is non-GRS type.
\end{theorem}
\begin{proof}
Without loss of generality, we assume $ \vv = \boldsymbol{1} $.  
For any integer $ i \geqslant 0 $, let $ \va^i = (\alpha_1^i, \dots, \alpha_n^i) $. 
For $ 3 \leqslant k < \frac{n+2}{2}$, by \cref{lem: GRS2}, we only need to show that the dimension of the Schur product $ C_1^2 $ is not equal to $ 2k-1$. 
From the definition of $ C_1 $ and \cref{lem: Schurprod_C}, 
we have
\begin{align*}
C_1^2  
= & \langle (\va^{i+j}, 0), \va^i (\va^h + \eta \va^{k+1}, 0), (\va^h + \eta \va^{k+1}, 1)^2 \rangle  \\ 
= & \langle (\va^s, 0), (\va^{i+h} + \eta \va^{i+k+1}, 0), (\va^{2h} + 2\eta \va^{h+k+1} + \eta^2 \va^{2k+2}, 1) \rangle 
\end{align*}
where $ 0 \leqslant s \leqslant 2k-2 $, $ 0 \leqslant i, j \leqslant k-1 $ and $ i, j \ne h $.
We analyze the following cases.
\begin{enumerate}[(i)]
\item If $ h = 0 $, we have $ 1 \leqslant i \leqslant k-1 $, then $ C_1^2 $ can be expressed as
    \begin{align*}
    C_1^2 = \langle 
    & (\va^s, 0), (\va + \eta \va^{k+2}, 0), \dots, (\va^{k-1} + \eta \va^{2k}, 0), \\
    & (\va^0 + 2\eta \va^{k+1} + \eta^2 \va^{2k+2}, 1) \colon 0 \leqslant s \leqslant 2k-2 \rangle.
    \end{align*}
    Furthermore, since $ k \geqslant 3 $, we have $ 2k-1 \geqslant k+2 $, and then
    \[
       C_1^2 =  \langle (\va^t, 0), (2\eta \va^{k+1} + \eta^2 \va^{2k+2}, 1) \colon 0 \leqslant t \leqslant 2k \rangle.
    \] 
\item If $ 1 \leqslant h \leqslant k-2 $, we get $ 0 \leqslant i \leqslant k-1 $ and $ i \ne h $, 
    then $ C_1^2 $ can be written as
    \begin{align*}
    C_1^2 = \langle 
    & (\va^s, 0), (\va^{h} + \eta \va^{k+1}, 0), \dots, (\va^{2h-1} + \eta \va^{h+k}, 0), \\
    & (\va^{2h+1} + \eta \va^{h+k+2}, 0), \dots, (\va^{h+k-1} + \eta \va^{2k}, 0), \\
    & (\va^{2h} + 2\eta \va^{h+k+1} + \eta^2 \va^{2k+2}, 1) \colon 0 \leqslant s \leqslant 2k-2 \rangle.
    \end{align*}
    If $ 2k - 1 = h + k + 1 $, then $ h = k - 2 $, so $ C_1^2 $ can be further simplified to
    \[
    C_1^2 = \left \{
    \begin{array}{lc}
    \langle (\va^{t}, 0), (2\eta \va^{h+k+1} + \eta^2 \va^{2k+2}, 1) \rangle, & 1 \leqslant h \leqslant k-3, \\
    \langle (\va^{\ell}, 0), (2\eta \va^{2k-1} + \eta^2 \va^{2k+2}, 1) \rangle, & h = k-2,
    \end{array}
    \right. 
    \]
    where $ 0 \leqslant t \leqslant 2k $, $ 0 \leqslant \ell \leqslant 2k $ and $ \ell \ne 2k-1 $.
\end{enumerate}
Combining the above two cases, the Schur product $ C_1^2 $ can be expressed as
\begin{align*}
C_1^2 = \left \{
\begin{array}{lc}
\langle (\va^{t}, 0), (2\eta \va^{h+k+1} + \eta^2 \va^{2k+2}, 1) \rangle, & 0 \leqslant h \leqslant k-3, \\
\langle (\va^{\ell}, 0), (2\eta \va^{2k-1} + \eta^2 \va^{2k+2}, 1) \rangle, & h = k-2, 
\end{array}
\right.
\end{align*}
where $ 0 \leqslant t \leqslant 2k $, $ 0 \leqslant \ell \leqslant 2k $ and $ \ell \ne 2k-1 $.

Next, we shall determine the dimension of $C_1^2$ by considering the parity of  $ q $.
If $ q $ is odd, then
\begin{align*}
C_1^2 = 
\langle (\va^{t}, 0), (\va^{2k+2}, \eta^{-2}) \colon 0 \leqslant t \leqslant 2k \rangle.
\end{align*}
Since $ k \leqslant \frac{n+1}{2}$, we get $ 2k-1 \leqslant n $. Thus the following matrix
\[
\begin{pmatrix}
  1 & \dots & 1 & 0 \\
  \alpha_1 & \dots &  \alpha_n & 0 \\
  \vdots &  & \vdots & \vdots \\ 
  \alpha_1^{2k-1} & \dots & \alpha_n^{2k-1} & 0 \\
  \alpha_1^{2k} & \dots & \alpha_n^{2k} & 0 \\
  \alpha_1^{2k+2} & \dots & \alpha_n^{2k+2} & \eta^{-2} \\
\end{pmatrix}
\]
has a $ 2k \times 2k $ invertible submatrix as follows:
\[
\begin{pmatrix}
  1 & \dots & 1 & 0 \\
  \alpha_1 & \dots &  \alpha_{2k-1} & 0 \\
  \vdots &  & \vdots & \vdots \\ 
  \alpha_1^{2k-2} & \dots & \alpha_{2k-1}^{2k-2} & 0 \\
  \alpha_1^{2k+2} & \dots & \alpha_{2k-1}^{2k+2} & \eta^{-2} \\
\end{pmatrix}.
\]
So $ \dim(C_1^2) \geqslant  2k > 2k-1$.
Otherwise, if $ q $ is even, then
\begin{align*}
C_1^2 & = 
\begin{cases}
\langle (\va^{t}, 0), (\va^{2k+2}, \eta^{-2}) \rangle, & 0 \leqslant h \leqslant k-3, \\
\langle (\va^{\ell}, 0), (\va^{2k+2}, \eta^{-2}) \rangle, & h = k-2, 
\end{cases}
\end{align*}
where $ 0 \leqslant t \leqslant 2k $, $ 0 \leqslant \ell \leqslant 2k $ and $ \ell \ne 2k-1 $.
Similar to the case of odd $ q $, we also obtain $ \dim(C_1^2) \geqslant  2k > 2k-1$.

Thus, $ C_1 $ is non-GRS type for $ 3 \leqslant k < \frac{n+2}{2}$. 
\end{proof}

In the following, we will indicate that $ \mathcal{C} $ is non-GRS type.
\begin{theorem}\label{thm: CnonGRS}
If $ 3 \leqslant k \leqslant n-3 $, then $ \mathcal{C} $ is non-GRS type.
\end{theorem}
\begin{proof}
Up to the equivalence of codes, we only need to consider $ \vv = \boldsymbol{1} $. 
Then we analyze two cases according to the range of $k$.

Firstly, we consider the case of $ 3 \leqslant k < \frac{n+3}{2}$.  
By \cref{thm: C=barC_1}, $ \mathcal{C} $ is an extended code of $ C_1 $. 
Therefore, by \cref{lem: barC_GRS_C}, to prove that $ \mathcal{C} $ is non-GRS type, 
it suffices to show that $ C_1 $ is non-GRS type. 
From \cref{thm: C_1noGRS}, this holds when $ 3 \leqslant k < \frac{n+2}{2} $. 
Since $ \mathcal{C} $ has length one greater than $ C_1 $, 
we conclude that $ \mathcal{C} $ is non-GRS type.

Secondly, we show that $ \mathcal{C} $ is non-GRS type for $ \frac{n+3}{2} \leqslant k \leqslant n-3 $.
Let $ u_i = \prod\limits_{j = 1, j \ne i}^{n} (\alpha_i - \alpha_j)^{-1} $ for $ i = 1, \ldots, n $.
By \cref{lem: (Van-M)u=1}, we have $ G \vc_i^{\T} = \boldsymbol{0}$ for $ i = 1, 2, 3$,
where
\begin{align*}
\vc_1 & = (
u_1 \alpha_1^{n-k-3}, \dots, u_n \alpha_n^{n-k-3}, 0, 0), \\
\vc_2 & = (u_1 \alpha_1^{n-k-2}, \dots, u_n \alpha_n^{n-k-2}, -\eta, 0),\\
\vc_3 & = (u_1 \alpha_1^{n-k-1}, \dots, u_n \alpha_n^{n-k-1}, - \eta \sum\limits_{i=1}^{n} \alpha_i, 0).
\end{align*}
Thus $ \vc_1, \vc_2, \vc_3 \in \mathcal{C}^{\perp}$, and so
\[ 
\vc = \vc_1 * \vc_3 - \vc_2^2 = (0, \dots, 0, -\eta^2, 0) \in (\mathcal{C}^{\perp})^2.
\]
This implies the minimum distance of $(\mathcal{C}^{\perp})^2 $ is $ d((\mathcal{C}^{\perp})^2) = 1 $. However, if $\mathcal{C}$ were an $[n+2, k] $ GRS code, then
\cref{prop: GRS^2} would imply  that $ d((\mathcal{C}^{\perp})^2)  = 2k-n $. 
Since $ \frac{n+3}{2} \leqslant k \leqslant n-3 $, we get $ d((\mathcal{C}^{\perp})^2) \geqslant 3 $, which is a contradiction. 
Thus $ \mathcal{C} $ is non-GRS type for $ \frac{n+3}{2} \leqslant k \leqslant n-3 $.

Combining both cases gives that $ \mathcal{C}$ is non-GRS type for $ 3 \leqslant k \leqslant n - 3 $.
\end{proof}

     
\section{The MDS and AMDS properties of $ \mathcal{C} $}\label{sec4}

From now on, without loss of generality, we assume that $ \vv = \boldsymbol{1} $.

Let $m,  n, r, h, k$ be integers with $ r \geqslant 0 $, $ 1 \leqslant m \leqslant n $ and $ 0 \leqslant h \leqslant k - 2 $, and $ \alpha_1, \alpha_2, \dots, \alpha_n $ be distinct elements in $ \fq $. 
For any subsets $ E , J \subseteq [n] $ with $ |E| = m $ and $ |J| = k - 1 $, we define the following notations for further use:
\begin{align}\label{eq:Sr}
S_r(E) = \begin{cases}
1, & r = 0, \\
(-1)^r \sum\limits_{i_1 < \dots < i_r \in E }  \alpha_{i_1} \alpha_{i_2} \cdots \alpha_{i_r}, & 1 \leqslant r \leqslant m, \\
0, & r > m, 
\end{cases}
\end{align}
and
\begin{align}\label{eq:Delta}  
\Delta  = S_{k-h-1}(J) \big(S_1(J)^2 - S_2(J) \big) - S_1(J)S_{k-h}(J) + S_{k-h+1}(J). 
\end{align}

\subsection{MDS property of $ \mathcal{C} $}

The linear code $ \mathcal{C} $ is not necessarily MDS, however, we have the following theorem.
\begin{theorem}\label{th: MDSh>0}
Let $ k, n, h $ be integers with $ 3 \leqslant k < k + 1 < n \leqslant q $ and $ 0 \leqslant h \leqslant k-2 $. Then the code $ \mathcal{C} $ is MDS if and only if the following four conditions hold simultaneously:
\begin{enumerate}[\upshape(1)]
    \item \label{MDSh>0I} $ \eta^{-1} \ne S_{k-h+1}(I) - S_1(I) S_{k-h}(I) $ for any subset $ I \subseteq [n] $ with $ |I| = k $.
    \item \label{MDSh>0J} $ S_{k-h-1}(J) \ne 0 $ for any subset $ J \subseteq [n] $ with $ |J| = k-1 $.
    \item $ S_{k-h-1}(J) + \delta - \delta \eta \Delta \neq 0 $ for any subset $ J \subseteq [n] $ with $ |J| = k-1 $.
    \item $ S_{k-h-2}(L) \ne 0 $ for any subset $ L \subseteq [n] $ with $ |L| = k-2 $.
\end{enumerate}
\end{theorem}

\begin{proof}
By \cref{lem: MDS}, the code $ \mathcal{C} $ is MDS if and only if any $ k $ columns of $ G $ are linearly independent, equivalently, any submatrix formed by selecting $ k $ columns from $ G $ is full rank. 
We analyze four cases based on the selection of columns.
\begin{enumerate}[\upshape(1)]
\item Assume that the $ k $ selected columns come from the first $ n $ columns, then they form a submatrix with the following form:
\begin{align*}
	B_1
	& = \begin{pmatrix}
	1 & \dots & 1  \\
	\alpha_{i_1} & \dots & \alpha_{i_k}  \\
	\vdots &  &  \vdots  \\
	\alpha_{i_1}^{h-1} & \dots & \alpha_{i_k}^{h-1}  \\
	\alpha_{i_1}^h + \eta \alpha_{i_1}^{k+1} & \dots & \alpha_{i_k}^h + \eta \alpha_{i_k}^{k+1}  \\
	\alpha_{i_1}^{h+1} & \dots &  \alpha_{i_k}^{h+1}  \\
	\vdots &  & \vdots   \\
	\alpha_{i_1}^{k-1} & \dots & \alpha_{i_k}^{k-1} 
	\end{pmatrix},
\end{align*}
where $ I = \{i_1, i_2, \dots, i_k \} \subseteq [n]$. 
Then
\begin{align*}
	\det(B_1) 
	 = \prod_{i,j \in I \atop i < j} (\alpha_j - \alpha_i) +
	(-1)^{k-h-1} \eta \cdot 
    \begin{vmatrix}
	1 & 1 & \dots & 1  \\
	\alpha_{i_1} & \alpha_{i_2} & \dots & \alpha_{i_k}  \\
	\vdots & \vdots &  &  \vdots  \\
	\alpha_{i_1}^{h-1} & \alpha_{i_2}^{h-1} & \dots & \alpha_{i_k}^{h-1}  \\
	\alpha_{i_1}^{h+1} & \alpha_{i_2}^{h+1} & \dots & \alpha_{i_k}^{h+1}  \\
	\vdots & \vdots &  & \vdots   \\
	\alpha_{i_1}^{k-1} & \alpha_{i_2}^{k-1} & \dots & \alpha_{i_k}^{k-1} \\
	\alpha_{i_1}^{k+1} & \alpha_{i_2}^{k+1} & \dots & \alpha_{i_k}^{k+1}  
	\end{vmatrix}.
\end{align*}
Denote by $A$ the determinant 
\[
    \begin{vmatrix}
	1 & 1 & \dots & 1  \\
	\alpha_{i_1} & \alpha_{i_2} & \dots & \alpha_{i_k}  \\
	\vdots & \vdots &  &  \vdots  \\
	\alpha_{i_1}^{h-1} & \alpha_{i_2}^{h-1} & \dots & \alpha_{i_k}^{h-1}  \\
	\alpha_{i_1}^{h+1} & \alpha_{i_2}^{h+1} & \dots & \alpha_{i_k}^{h+1}  \\
	\vdots & \vdots &  & \vdots   \\
	\alpha_{i_1}^{k-1} & \alpha_{i_2}^{k-1} & \dots & \alpha_{i_k}^{k-1} \\
	\alpha_{i_1}^{k+1} & \alpha_{i_2}^{k+1} & \dots & \alpha_{i_k}^{k+1}  
	\end{vmatrix}.
\]
To calculate $ A $, we only need to calculate the following Vandermonde determinant: 
\[  A(x, y) = 
    \begin{vmatrix}	
    1 & \dots & 1 & 1 & 1 \\
	\alpha_{i_1} & \dots & \alpha_{i_k} & x & y  \\
	\vdots &  &  \vdots & \vdots & \vdots \\
	\alpha_{i_1}^{h-1} & \dots & \alpha_{i_k}^{h-1} & x^{h-1} & y^{h-1}  \\
	\alpha_{i_1}^{h} & \dots & \alpha_{i_k}^{h} & x^{h} & y^{h}  \\
	\alpha_{i_1}^{h+1} & \dots & \alpha_{i_k}^{h+1} & x^{h+1} & y^{h+1}  \\
	\vdots &  & \vdots  &  \vdots  &  \vdots \\
	\alpha_{i_1}^{k-1} & \dots & \alpha_{i_k}^{k-1} & x^{k-1} & y^{k-1} \\
	\alpha_{i_1}^{k} & \dots & \alpha_{i_k}^{k} & x^{k} & y^{k} \\
	\alpha_{i_1}^{k+1} & \dots & \alpha_{i_k}^{k+1} & x^{k+1} & y^{k+1} 
	\end{vmatrix},
\]
where $ x \ne y \ne \alpha_j $, $ j = i_1, \dots, i_k $. 
It follows that 
\begin{align}
      A(x, y)
 & = (y - x) \prod_{j=1}^{k} (y-\alpha_{i_j}) (x-\alpha_{i_j}) \prod_{1 \leqslant s < t \leqslant      k} (\alpha_{i_t} - \alpha_{i_s}) \label{eq: A(x,y)-1} \\
 & = \cdots + (-1)^{h+k+1} A x^h y^k + \cdots. \label{eq: A(x,y)-2}
\end{align} 
By comparing the coefficients of $ x^h y^k $ in \eqref{eq: A(x,y)-1} and \eqref{eq: A(x,y)-2}, we can obtain 
\begin{align*}
A 
   = & (-1)^{h+k+1} \prod_{s,t \in I \atop s < t} (\alpha_t - \alpha_s) \big(S_1(I) S_{k-h}(I) - S_{k-h+1}(I) \big),
\end{align*}
where $ I = \{i_1, \dots, i_k\} \subseteq [n]$.
Thus we conclude that
\[
   \det(B_1) = \prod_{s,t \in I \atop s < t} (\alpha_t - \alpha_s) \bigg(1 - \eta \Big(S_{k-h+1}(I) - S_1(I) S_{k-h}(I) \Big) \bigg).
\]
It follows that the submatrix $ B_1 $ has full rank if and only if 
\[ 
   \eta^{-1} \ne S_{k-h+1}(I) - S_1(I) S_{k-h}(I) 
\]
for any subset $ I \subseteq [n] $ with $ |I| = k $.

\item Assume that the $ k $ selected columns are formed by combining any $ k - 1 $ columns from the first $ n $ columns with the $ (n+1) $th column, the submatrix has the following form: 
\begin{align*}
	B_2 
	& = \begin{pmatrix}
	1 & \dots & 1 & 0 \\
	\alpha_{i_1} & \dots & \alpha_{i_{k-1}} & 0 \\
	\vdots &  &  \vdots &  \vdots  \\
	\alpha_{i_1}^{h-1} & \dots & \alpha_{i_{k-1}}^{h-1} & 0 \\
	\alpha_{i_1}^h + \eta \alpha_{i_1}^{k+1} & \dots & \alpha_{i_{k-1}}^h + \eta \alpha_{i_{k-1}}^{k+1} & 1 \\
	\alpha_{i_1}^{h+1} & \dots &  \alpha_{i_{k-1}}^{h+1} & 0 \\
	\vdots &  & \vdots & \vdots \\
	\alpha_{i_1}^{k-1} & \dots & \alpha_{i_{k-1}}^{k-1} & 0
	\end{pmatrix}.
\end{align*}	
By a similar calculation to case (1), we get 
\[
   \det(B_2) = \prod_{s,t \in J \atop s < t} (\alpha_t - \alpha_s) S_{k-h-1}(J),
\]
where $ J = \{i_1, i_2, \dots, i_{k-1} \} \subseteq [n]$. 
Thus the submatrix $ B_2 $ has full rank if and only if 
$ S_{k-h-1}(J) \ne 0 $ for any subset $ J \subseteq [n] $ with $ |J| = k-1 $.

\item  Assume that the $ k $ selected columns have the following form:
\begin{align*}
	B_3 
	& = \begin{pmatrix}
	1 & \dots & 1 & 0 \\
	\alpha_{i_1} & \dots & \alpha_{i_{k-1}} & 0 \\
	\vdots &  &  \vdots &  \vdots  \\
	\alpha_{i_1}^{h-1} & \dots & \alpha_{i_{k-1}}^{h-1} & 0 \\
	\alpha_{i_1}^h + \eta \alpha_{i_1}^{k+1} & \dots & \alpha_{i_{k-1}}^h + \eta \alpha_{i_{k-1}}^{k+1} & 1 \\
	\alpha_{i_1}^{h+1} & \dots &  \alpha_{i_{k-1}}^{h+1} & 0 \\
	\vdots &  & \vdots & \vdots \\
	\alpha_{i_1}^{k-1} & \dots & \alpha_{i_{k-1}}^{k-1} & \delta
	\end{pmatrix}.
\end{align*}	
Then
\begin{align*}
   \det(B_3) = \prod_{s,t \in J \atop s < t} (\alpha_t - \alpha_s)   \Big(S_{k-h-1}(J) + \delta - \delta \eta \Delta  \Big),
\end{align*}
where 
$ \Delta  = S_{k-h-1}(J) \big(S_1(J)^2 - S_2(J) \big) + S_{k-h+1}(J) - S_1(J)S_{k-h}(J) $       
and $ J = \{i_1, i_2, \dots, i_{k-1} \} \subseteq [n] $.
Thus the submatrix $ B_3 $ has full rank if and only if 
\begin{align*}
    S_{k-h-1}(J) + \delta - \delta \eta \Delta \neq 0
\end{align*}
for any subset $ J \subseteq [n] $ with $ |J| = k-1 $.

\item Assume that the $ k $ selected columns are formed by combining any $ k - 2 $ columns from the first $ n $ columns with the $ (n + 1) $th column and the $ (n + 2) $th column, then the submatrix has the following form:
\begin{align*}
	B_4 
	& = \begin{pmatrix}
	1 & \dots & 1 & 0 & 0 \\
	\alpha_{i_1} & \dots & \alpha_{i_{k-2}} & 0 & 0\\
	\vdots &  &  \vdots &  \vdots &  \vdots \\
	\alpha_{i_1}^{h-1} & \dots & \alpha_{i_{k-2}}^{h-1} & 0 & 0 \\
	\alpha_{i_1}^h + \eta \alpha_{i_1}^{k+1} & \dots & \alpha_{i_{k-2}}^h + \eta \alpha_{i_{k-2}}^{k+1} & 1 & 1 \\
	\alpha_{i_1}^{h+1} & \dots &  \alpha_{i_{k-2}}^{h+1} & 0 & 0 \\
	\vdots &  & \vdots & \vdots & \vdots \\
	\alpha_{i_1}^{k-1} & \dots & \alpha_{i_{k-2}}^{k-1} & 0 & \delta
	\end{pmatrix} 
\end{align*}
Then
\[ 
\det(B_4) = \delta \prod_{s,t \in L \atop s < t} (\alpha_t - \alpha_s) S_{k-h-2}(L),
\]
where $ L = \{i_1, i_2, \dots, i_{k-2} \} \subseteq [n]$.
Thus the submatrix $ B_4 $ has full rank if and only if $ S_{k-h-2}(L) \ne 0 $ 
for any subset $ L \subseteq [n] $ with $ |L| = k-2 $.
\end{enumerate}
Combining all cases, the code $ \mathcal{C} $ is MDS if and only if Conditions $ (1)-(4) $ hold.
\end{proof}

As a special case of \cref{th: MDSh>0}, we have the following result for $h=0$.
\begin{corollary}\label{th: MDSh=0}
Let $ k, n $ be integers with $ 3 \leqslant k < k + 1 < n \leqslant q $. 
Then the code $ \mathcal{C} $ is MDS if and only if the following four conditions hold simultaneously:
\begin{enumerate}[\upshape(1)]
    \item \label{MDSh=0I} $ \eta^{-1} \ne -S_1(I) S_k(I) $ for any subset $ I \subseteq [n] $ with $ |I| = k $.
    \item \label{MDSh=0J} $ S_{k-1}(J) \ne 0 $ for any subset $ J \subseteq [n] $ with $ |J| = k-1 $.
    \item $ S_{k-1}(J) + \delta - \delta \eta S_{k-1}(J) \big(S_1(J)^2 - S_2(J) \big) \ne 0 $ for any subset $ J \subseteq [n] $ with $ |J| = k-1 $. 
    \item $ S_{k-2}(L) \ne 0 $ for any subset $ L \subseteq [n] $ with $ |L| = k-2 $.
\end{enumerate}
\end{corollary}

We give the following examples of non-GRS type MDS codes, which are all confirmed by Magma programs.

\begin{example}\label{eg:3.4}
Let $ q = 11 $, $ n = 6 $, $ k = 3 $, $ h = 1 $ and $ \va = (0, 1, 2, 3, 4, 5) $. 
Taking $ (\eta, \delta) = (4, 7) $, \cref{th: MDSh>0} implies that $ \mathcal{C} $ is a non-GRS type MDS code with parameters $ [8, 3, 6] $. 
\end{example}

\begin{example}\label{eg:4.7}
Let $ \xi $ be a primitive element of $ \mathbb{F}_{16} $. 
Let $ n = 7 $, $ k = 4 $, $ h = 2 $, $ \va = (0, \xi, \xi^2, \xi^4, \xi^6, \xi^7, \xi^{13}) $. 
Taking $ (\eta, \delta) = (\xi, \xi^7) $ gives that $ \mathcal{C} $ is a non-GRS type MDS code with parameters $ [9, 4, 6] $ from \cref{th: MDSh>0}. 
Magma verifies two additional $ (\eta, \delta) $ pairs listed in \cref{tab:MDS} that yield the same parameters.
\end{example}

\begin{example}\label{eg:3.6}
Let $ q = 19 $, $ n = 8 $, $ k = 5 $, $ h = 0 $ and $ \va = (3, 4, 5, 6, 13, 14, 15, 16) $. 
Taking $ (\eta, \delta) = (15, 6) $ or $ (15, 18) $, \cref{th: MDSh=0} implies that $ \mathcal{C} $ is a non-GRS type MDS code with parameters $ [10, 5, 6] $. 
\end{example}

\begin{table}[h] 
\centering
\caption{MDS codes obtained from \cref{th: MDSh>0} and \cref{th: MDSh=0}}  
\label{tab:MDS}
\begin{tabular}{ccccccc}
\toprule
MDS code  &  $ q $  &  $ h $  &  $ \va $  &  $ (\eta, \delta) $  & \\
\midrule
$ [8,3,6] $  & 11  &  1  &  $ (0, 1, 2, 3, 4, 5) $  &  $ (4, 7) $  &  \cref{eg:3.4}  \\
\midrule
$ [9,4,6] $  &  16  &  2  &  $ (0, \xi, \xi^2, \xi^4, \xi^6, \xi^7, \xi^{13}) $  &  \makecell[c]{$ (\xi, \xi^7) $, \\$ (\xi^2, \xi^5) $, $ (\xi^{12}, \xi) $}  &  \cref{eg:4.7}  \\
\midrule
$ [10,5,6] $  &  19  &  0  &  $ (3, 4, 5, 6, 13, 14, 15, 16) $  &  $ (15, 6) $, $ (15, 18) $  &  \cref{eg:3.6}  \\
\bottomrule
\end{tabular}
\end{table}

\subsection{AMDS property of $ \mathcal{C} $ }

Now, we consider when the code $ \mathcal{C} $ becomes AMDS. 
Clearly, AMDS codes are non-GRS type.

\begin{theorem}\label{th: AMDSh>0}
Let $ k, n, h $ be integers with $ 3 \leqslant k < n-1 $ and $ 0 \leqslant h \leqslant k-2 $. 
Then the code $ \mathcal{C} $ is AMDS if and only if the following three conditions hold simultaneously:
\begin{enumerate}[\upshape(1)]
\item For any subset $ M \subseteq [n] $ with $ | M | = k + 1 $, there exists a subset $ I \subset M $ with $ | I | = k $ such that 
\[ 
\eta^{-1} \ne S_{k-h+1} (I) - S_{1} (I) S_{k-h} (I).
\]
\item For any subset $ I \subseteq [n] $ with $ | I | = k $, there exist a subset $ J \subset I $ with $ | J | = k - 1 $ such that 
\[ 
S_{k-h-1} (J) + \delta - \delta \eta \Delta \ne 0. 
\]
\item One of the following conditions holds:
\begin{enumerate}[\upshape(a)]
\item There exists a subset $ I \subseteq [n] $ with $ |I| = k $ such that 
\[
\eta^{-1} = S_{k-h+1}(I) - S_1(I) S_{k-h}(I).
\]
\item There exists a subset $ J \subseteq [n] $ with $ |J| = k-1 $ such that $ S_{k-h-1}(J) = 0 $.
\item There exists a subset $ J \subseteq [n] $ with $ |J| = k-1 $ such that 
\[
S_{k-h-1}(J) + \delta - \delta \eta \Delta = 0.
\] 
\item There exists a subset $ L \subseteq [n] $ with $ |L| = k-2 $ such that $ S_{k-h-2}(L) = 0 $. 
\end{enumerate}
\end{enumerate}
\end{theorem}
\begin{proof}
By \cref{lem: AMDS}, the code $ \mathcal{C} $ is AMDS if and only if there exist $ k $ columns in $G$ with rank at most $ k - 1$ and  any $ k + 1 $ columns of $G$ have rank exactly $ k $. 
The first part holds if and only if one of the following conditions holds:
\begin{enumerate}[\upshape(a)]
\item There exists a subset $ I \subseteq [n] $ with $ |I| = k $ such that 
\[
\eta^{-1} = S_{k-h+1}(I) - S_1(I) S_{k-h}(I).
\]
\item There exists a subset $ J \subseteq [n] $ with $ |J| = k-1 $ such that $ S_{k-h-1}(J) = 0 $.
\item There exists a subset $ J \subseteq [n] $ with $ |J| = k-1 $ such that 
\[
S_{k-h-1}(J) + \delta - \delta \eta \Delta = 0.
\] 
\item There exists a subset $ L \subseteq [n] $ with $ |L| = k-2 $ such that $ S_{k-h-2}(L) = 0 $. 
\end{enumerate}

Next, we consider the second part. Any $ k + 1 $ columns of $G$ have rank exactly $ k $ if and only if any $ k \times (k+1) $ submatrix of $ G $ has rank exactly $ k $. We analyze the following four cases.
\begin{enumerate}[\upshape(1)]
\item Assume that we take any $ k + 1$ columns from the first $n$ columns of $G$ to form a $ k \times (k+1) $ submatrix of $ G $, then this submatrix has the following form:
\begin{align*}
	D_1 
	& = \begin{pmatrix}
	1 & \dots & 1  \\
	\alpha_{i_1} & \dots & \alpha_{i_{k+1}}  \\
	\vdots &  &  \vdots  \\
	\alpha_{i_1}^{h-1} & \dots & \alpha_{i_{k+1}}^{h-1}  \\
	\alpha_{i_1}^h + \eta \alpha_{i_1}^{k+1} & \dots & \alpha_{i_{k+1}}^h + \eta \alpha_{i_{k+1}}^{k+1}  \\
	\alpha_{i_1}^{h+1} & \dots &  \alpha_{i_{k+1}}^{h+1}  \\
	\vdots &  & \vdots   \\
	\alpha_{i_1}^{k-1} & \dots & \alpha_{i_{k+1}}^{k-1} 
	\end{pmatrix}, 	
\end{align*}
where $ M = \{i_1, i_2, \dots, i_{k+1} \} \subseteq [n] $. 
Let $ D_1(j) $ be the matrix deleting the $j$th column from $ D_1 $, then $ \rk(D_1) = k $ if and only if there exists a matrix $ D_1(j) $ such that $ \det(D_1(j)) \ne 0 $, where $ 1 \leqslant j \leqslant k+1 $. 
We can suppose $ j = k + 1 $ without loss of generality. 
Note that $ B_1 = D_1(k+1) $. By \cref{th: MDSh>0},
\[
\det(B_1) \ne 0 \quad \text{if and only if} \quad \eta^{-1} \ne S_{k-h+1} (I) - S_{1} (I) S_{k-h} (I).
\]
Then $ \rk(D_1) = k $ if and only if for any subset $ M \subseteq [n] $ with $ | M | = k + 1 $, there exists a subset $ I \subset M $ with $ | I | = k $ such that 
\[ 
\eta^{-1} \ne S_{k-h+1} (I) - S_{1} (I) S_{k-h} (I).
\]
\item Assume the $ k \times (k+1) $ submatrix of $ G $ consists of any $ k $ columns from the first $ n $ columns of $G$ and the $(n + 1)$th column of $ G $, then this submatrix has the following form:
\begin{align*}
	D_2 
	& = \begin{pmatrix}
	1 & \dots & 1 & 0 \\
	\alpha_{i_1} & \dots & \alpha_{i_{k}} & 0 \\
	\vdots &  &  \vdots &  \vdots  \\
	\alpha_{i_1}^{h-1} & \dots & \alpha_{i_{k}}^{h-1} & 0 \\
	\alpha_{i_1}^h + \eta \alpha_{i_1}^{k+1} & \dots & \alpha_{i_{k}}^h + \eta \alpha_{i_{k}}^{k+1} & 1 \\
	\alpha_{i_1}^{h+1} & \dots &  \alpha_{i_{k}}^{h+1} & 0 \\
	\vdots &  & \vdots & \vdots   \\
	\alpha_{i_1}^{k-1} & \dots & \alpha_{i_{k}}^{k-1} & 0
	\end{pmatrix}, 	
\end{align*}
where $ I = \{i_1, i_2, \dots, i_{k} \} \subseteq [n] $.
Denote the row vectors of $ D_2 $ by $ \beta_0 $, $ \beta_1, \ldots, \beta_{k-1} $. 
Since all rows of the Vandermonde matrix are linearly independent, the $ k $ rows of the matrix $ D_2 $ are linearly independent, thus the rank of $ D_2 $ is $ k $. 

\item Assume that the columns of the $ k \times (k+1) $ submatrix of $ G $ are formed by any $ k $ columns from the first $ n $ columns of $ G $, plus the $(n + 2)$th column of $ G $, then this submatrix has the following form:
\begin{align*}
	D_3 
	& = \begin{pmatrix}
	1 & \dots & 1 & 0 \\
	\alpha_{i_1} & \dots & \alpha_{i_{k}} & 0 \\
	\vdots &  &  \vdots &  \vdots  \\
	\alpha_{i_1}^{h-1} & \dots & \alpha_{i_{k}}^{h-1} & 0 \\
	\alpha_{i_1}^h + \eta \alpha_{i_1}^{k+1} & \dots & \alpha_{i_{k}}^h + \eta \alpha_{i_{k}}^{k+1} & 1 \\
	\alpha_{i_1}^{h+1} & \dots &  \alpha_{i_{k}}^{h+1} & 0 \\
	\vdots &  & \vdots &  \vdots   \\
	\alpha_{i_1}^{k-1} & \dots & \alpha_{i_{k}}^{k-1} & \delta
	\end{pmatrix}, 	
\end{align*}
where $ I = \{i_1, i_2, \dots, i_{k} \} \subseteq [n] $. 
Let $ D_3(j) $ be the matrix deleting the $j$th column from $ D_3 $, then $ \rk(D_3) = k $ if and only if there exists a matrix $ D_3(j) $ such that $ \det(D_3(j)) \ne 0 $, where $ 1 \leqslant j \leqslant k $. 
We can suppose $ j = k $ without loss of generality. Note that $ B_3 = D_3(k) $. By \cref{th: MDSh>0},
\[
\det(B_3) \ne 0 \quad \text{if and only if} \quad S_{k-h-1} (J) + \delta - \delta \eta \Delta \ne 0.
\]
Then $ \rk(D_3) = k $ if and only if for any subset $ I \subseteq [n] $ with $ | I | = k $, there exists a subset $ J \subset I $ with $ | J | = k - 1 $ such that 
\[ 
S_{k-h-1} (J) + \delta - \delta \eta \Delta \ne 0. 
\]
\item Assume that the columns of the $ k \times (k+1) $ submatrix of $ G $ consists of any $ k $ columns from the first $ n $ columns of $G$, plus the $(n + 1)$th column and the $(n + 2)$th column of $G$, then this submatrix has the following form:
\begin{align*}
	D_4 
	& = \begin{pmatrix}
	1 & \dots & 1 & 0 & 0 \\
	\alpha_{i_1} & \dots & \alpha_{i_{k-1}} & 0 & 0 \\
	\vdots &  &  \vdots & \vdots & \vdots \\
	\alpha_{i_1}^{h-1} & \dots & \alpha_{i_{k-1}}^{h-1} & 0 & 0 \\
	\alpha_{i_1}^h + \eta \alpha_{i_1}^{k+1} & \dots & \alpha_{i_{k-1}}^h + \eta \alpha_{i_{k-1}}^{k+1} & 1 & 1 \\
	\alpha_{i_1}^{h+1} & \dots &  \alpha_{i_{k-1}}^{h+1} & 0 & 0 \\
	\vdots &  & \vdots & \vdots & \vdots  \\
	\alpha_{i_1}^{k-1} & \dots & \alpha_{i_{k-1}}^{k-1} & 0 & \delta
	\end{pmatrix}, 	
\end{align*}
where $ L = \{i_1, i_2, \dots, i_{k-1} \} \subseteq [n] $. 
Similar to the proof of case (2), we have the $ k $ rows of the matrix $ D_4 $ are linearly independent, thus the rank of $ D_4 $ is $ k $.  
\end{enumerate}
This completes the proof. 
\end{proof}

Consequently, from \cref{th: AMDSh>0}, we have the following corollary for $h=0$.

\begin{corollary}\label{th: AMDSh=0}
Let $ k, n, h $ be integers with $ 3 \leqslant k < n-1 $. 
Then the code $ \mathcal{C} $ is AMDS if and only if the following three conditions hold simultaneously:
\begin{enumerate}[\upshape(1)]
\item For any subset $ M \subseteq [n] $ with $ | M | = k + 1 $, there exists a subset $ I \subset M $ with $ | I | = k $ such that 
\[ 
\eta^{-1} \ne - S_{1} (I) S_{k} (I).
\]
\item For any subset $ I \subseteq [n] $ with $ | I | = k $, there exists a subset $ J \subset I $ with $ | J | = k - 1 $ such that 
\[ 
S_{k-1} (J) + \delta - \delta \eta S_{k-1}(J) \big(S_1(J)^2 - S_2(J) \big) \ne 0. 
\]
\item One of the following conditions holds:
\begin{enumerate} [\upshape(a)]
\item There exists a subset $ I \subseteq [n] $ with $ |I| = k $ such that 
\[
\eta^{-1} = -S_1(I) S_k(I).
\]
\item There exists a subset $ J \subseteq [n] $ with $ |J| = k-1 $ such that $ S_{k-1}(J) = 0 $.
\item There exists a subset $ J \subseteq [n] $ with $ |J| = k-1 $ such that 
\[
S_{k-1}(J) + \delta - \delta \eta S_{k-1}(J) \big(S_1(J)^2 - S_2(J) \big) = 0. 
\] 
\item There exists a subset $ L \subseteq [n] $ with $ |L| = k-2 $ such that $ S_{k-2}(L) = 0 $. 
\end{enumerate}
\end{enumerate}   
\end{corollary}

All examples below are verified using Magma programs.
\begin{example}\label{eg:3.7}
Let $ q = 5 $, $ n = 5 $, $ k = 3 $, $ h = 1 $ and $ \va = (0, 1, 2, 3, 4) $. 
If we take $ (\eta, \delta) = (1, 1) $, by \cref{th: AMDSh>0}, 
$ \mathcal{C} $ is an AMDS code with parameters $ [7, 3, 4] $. 
Magma verifies $11$ additional $ (\eta, \delta) $ pairs listed in \cref{tab:AMDS} that yield the same parameters.
\end{example}

\begin{example}\label{eg:3.8}
Let $ \gamma $ be a  primitive element of $ \mathbb{F}_8 $. 
Let $ n = 6 $, $ k = 4 $, $ h = 2 $ and $ \va = (0, 1, \gamma, \gamma^2, \gamma^3, \gamma^5) $. 
Taking $ (\eta, \delta) = (1, \gamma) $, \cref{th: AMDSh>0} gives $ \mathcal{C} $ is an AMDS code with parameters $ [8, 4, 4] $. 
Furthermore, there are $32$ additional such $(\eta, \delta)$ pairs listed in \cref{tab:AMDS} .
\end{example}

\begin{example}\label{eg:3.9}
Let $ q = 7 $, $ n = 5 $, $ k = 3 $, $ h = 0 $ and $ \va = (2, 3, 4, 5, 6) $. 
Put $ (\eta, \delta) = (5, 2) $. 
Then $ \mathcal{C} $ is a $ [7, 3, 4] $ AMDS code according to \cref{th: AMDSh=0}. 
Additionally, there are totally $25 $ such $(\eta, \delta)$  pairs listed in \cref{tab:AMDS}.
\end{example}

\begin{table}[h] 
\centering
\caption{AMDS codes obtained from \cref{th: AMDSh>0} and \cref{th: AMDSh=0}} 
\label{tab:AMDS}
\begin{tabular}{ccccccc}
\toprule
AMDS code  &  $ q $  &  $ h $  &  $ \va $  &  $ (\eta, \delta) $  &  \\
\midrule
$ [7,3,4] $  &  7  &  0  &  $ (2, 3, 4, 5, 6) $  &  \makecell[c]{(2, 1), (2, 2), (2, 3),\\ (2, 4), (2, 6), (3, 1),\\ (3, 2), (3, 3), (3, 4),\\ (3, 5) (4, 1), (4, 2),\\ (4, 4), (4, 5), (4, 6),\\ (5, 1), (5, 2), (5, 4),\\ (5, 5), (5, 6), (6, 1),\\ (6, 2), (6, 3), (6, 4), (6, 6)}  &  \cref{eg:3.9}   \\
\midrule
$ [7,3,4] $  &  5  &  1  &  $ (0, 1, 2, 3, 4) $  &  \makecell[c]{(1, 1), (1, 2), (1, 4),\\ (2, 2), (2, 3), (2, 4),\\ (3, 1), (3, 2), (3, 3),\\ (4, 1), (4, 3), (4, 4)}  &  \cref{eg:3.7}  \\
\midrule
$ [8,4,4] $  &  8  &  2  &  $ (0, 1, \gamma, \gamma^2, \gamma^3, \gamma^5) $  &  \makecell[c]{$ (1, \gamma) $, $ (1, \gamma^2) $, $ (1, \gamma^3) $,\\ $ (1, \gamma^4) $, $ (1, \gamma^5) $, $ (1, \gamma^6) $,\\ $ (\gamma^2, 1) $, $ (\gamma^2, \gamma) $, $ (\gamma^2, \gamma^2) $,\\ $ (\gamma^2, \gamma^3) $, $ (\gamma^2, \gamma^5) $, $ (\gamma^2, \gamma^6) $,\\ $ (\gamma^3, 1) $, $ (\gamma^3, \gamma) $, $ (\gamma^3, \gamma^3) $,\\ $ (\gamma^3, \gamma^5) $, $ (\gamma^4, 1) $, $ (\gamma^4, \gamma^2) $,\\ $ (\gamma^4, \gamma^3) $, $ (\gamma^4, \gamma^5) $, $ (\gamma^4, \gamma^6) $,\\ $ (\gamma^5, 1) $, $ (\gamma^5, \gamma) $, $ (\gamma^5, \gamma^2) $,\\ $ (\gamma^5, \gamma^4) $, $ (\gamma^5, \gamma^5) $, $ (\gamma^5, \gamma^6) $,\\ $ (\gamma^6, 1) $, $ (\gamma^6, \gamma) $, $ (\gamma^6, \gamma^2) $,\\ $ (\gamma^6, \gamma^3) $, $ (\gamma^6, \gamma^4) $, $ (\gamma^6, \gamma^6) $}  &  \cref{eg:3.8}  \\
\bottomrule
\end{tabular}
\end{table}


\section{The covering radii and deep holes of the codes}\label{sec5} 

This section investigates the covering radii and deep holes of linear codes, including a general result for AMDS codes. Note that the results for MDS codes are presented in \cref{lem: coveringradius} and \cref{lem: deephole2}.  We further apply them to characterize the covering radii and deep holes of the code $ \mathcal{C} $ defined in \cref{sec3}. Assume $ \vv = \boldsymbol{1} $ without loss of generality.

First, we present the covering radius and related deep hole of an AMDS code.
\begin{theorem}\label{lem: AMDS-CR}
Let $ G $ be a generator matrix of an $ [n, k] $ AMDS code $ C $. 
Then the covering radius of $ C $ is $ \rho(C) = n - k $ if and only if there exists a vector $ \vx \in \fqn $ such that the $ (k+1) \times n $ matrix $ \big( \frac{G}{\vx} \big) $ generates an $ [n, k+1] $ MDS code.
Moreover, $ \vx \in \fqn $ is a deep hole of $ C $ if and only if the $ (k+1) \times n $ matrix $ \big( \frac{G}{\vx} \big) $ generates an $ [n, k+1] $ MDS code.
\end{theorem}
\begin{proof}
The conclusion follows from \cref{def: coveringradius}, \cref{lem: SupercodeLemma}, and \cref{lem: RedundancyBound}.
\end{proof}

In the following theorem, we determine the covering radii and deep holes of the code $ \mathcal{C} $ when $ \mathcal{C} $ is MDS or AMDS by applying \cref{lem: coveringradius}, \cref{lem: deephole2}, or \cref{lem: AMDS-CR}. For simplicity, we define the notations as follows:
\begin{align*}
\Theta_1 & = 1-\eta \big(S_{k-h+1}(I) - S_1(I)S_{k-h}(I) \big), \\
\Theta_2 & = S_1(I) + \eta \big(S_2(I)S_{k-h}(I) - S_1(I)S_{k-h+1}(I) \big), \\
\Theta_3 & = S_{k-h-1}(J) + \delta - \delta \eta \Delta, \\
\Theta_4 & = S_1(J)S_{k-h-1}(J) - S_{k-h}(J),
\end{align*}
for any subset $ I $, $ J \subseteq [n] $ with $ | I | = k $ and $ | J | = k - 1 $, where $S_j(\cdot)$ and $ \Delta $ are defined as in \eqref{eq:Sr} and \eqref{eq:Delta}, respectively.

\begin{theorem}\label{th:crdp}
Let $ a, b \in \fq $ and let $ k, n, h $ be integers with $ 3 \leqslant k < k + 1 < n \leqslant q $ and $ 0 \leqslant h \leqslant k-2 $.
Assume $ \mathcal{C} $ is MDS or AMDS and the following conditions hold simultaneously: 
\begin{enumerate}[\upshape 1)]
\item $ \eta^{-1} \ne S_{k-h+1}(M) $ for any subset $ M \subseteq [n] $ with $ | M | = k+1 $.
\item $ S_{k-h}(I) + a \Theta_1 \ne 0 $ for any subset $ I \subseteq [n] $ with $ | I | = k $.
\item $ S_{k-h}(I) + b\Theta_1 + \delta \Theta_2 \ne 0 $ for any subset $ I \subseteq [n] $ with $ | I | = k $.
\item $ a \Theta_3 - b S_{k-h-1}(J) - \delta \Theta_4 \ne 0 $ for any subset $ J \subseteq [n] $ with $ | J | = k-1 $.
\end{enumerate}
Then the covering radius of $\mathcal{C}$ is $ \rho(\mathcal{C}) = n - k + 2 $ and the vector $ \vx = (\alpha_1^k, \dots, \alpha_n^k, a, b) $ is a deep hole of $ \mathcal{C} $.
\end{theorem}
\begin{proof}
The result follows by combining
\cref{lem: coveringradius}, \cref{lem: deephole2}, \cref{lem: MDS}, \cref{th: MDSh>0}, and \cref{lem: AMDS-CR}. 
\end{proof}

Taking $ a = b = 0 $ in \cref{th:crdp} leads to the following results.
\begin{corollary}\label{cor:a=b=0}
Let $ k, n, h $ be integers with $ 3 \leqslant k < k + 1 < n \leqslant q $ and $ 0 \leqslant h \leqslant k-2 $.
Assume that $ \mathcal{C} $ is MDS or AMDS and the following conditions hold simultaneously: 
\begin{enumerate}[\upshape 1)]
\item $ \eta^{-1} \ne S_{k-h+1}(M) $ for any subset $ M \subseteq [n] $ with $ | M | = k+1 $.
\item $ S_{k-h}(I) \ne 0 $ for any subset $ I \subseteq [n] $ with $ | I | = k $.
\item $ S_{k-h}(I) + \delta \Big(S_1(I) + \eta \big(S_{k-h}(I)S_2(I) - S_{k-h+1}(I)S_1(I) \big)\Big) \ne 0 $ for any subset $ I \subseteq [n] $ with $ | I | = k $.
\item $ S_{k-h-1}(J)S_1(J) - S_{k-h}(J) \ne 0 $ for any subset $ J \subseteq [n] $ with $ | J | = k-1 $.
\end{enumerate}
Then the covering radius of $\mathcal{C}$ is $ \rho(\mathcal{C}) = n - k + 2 $ and the vector $ \vx = (\alpha_1^k, \dots, \alpha_n^k, 0, 0) $ is a deep hole of $ \mathcal{C} $.
\end{corollary}

Taking $ h = 0 $ in \cref{th:crdp} yields the following corollary.
\begin{corollary}\label{cor:h=0}
Let $ a, b \in \fq $. Let $ k, n, h $ be integers with $ 3 \leqslant k < k + 1 < n \leqslant q $.
Assume that $ \mathcal{C} $ is MDS or AMDS and the following conditions hold simultaneously: 
\begin{enumerate}[\upshape 1)]
\item $ \eta^{-1} \ne S_{k+1}(M) $ for any subset $ M \subseteq [n] $ with $ | M | = k+1 $.
\item $ S_{k}(I) + a \big(1+\eta S_{k}(I)S_{1}(I) \big) \ne 0 $ for any subset $ I \subseteq [n] $ with $ | I | = k $.
\item $ S_{k}(I) + b \big(1+\eta S_{k}(I)S_{1}(I) \big) + \delta \big(S_1(I) + \eta S_2(I)S_{k}(I) \big) \ne 0 $ for any subset $ I \subseteq [n] $ with $ | I | = k $.
\item $ a \Big(S_{k-1}(J)+\delta-\delta \eta S_{k-1}(J) \big(S_1(J)^2-S_2(J) \big) \Big) - b S_{k-1}(J) - \delta S_{k-1}(J)S_1(J) \ne 0 $ for any subset $ J \subseteq [n] $ with $ | J | = k-1 $.
\end{enumerate}
Then the covering radius of $\mathcal{C}$ is $ \rho(\mathcal{C}) = n - k +2 $ and the vector $ \vx = (\alpha_1^k, \dots, \alpha_n^k, a, b) $ is a deep hole of $ \mathcal{C} $.
\end{corollary}

We give some examples for the above results, which are all verified using Magma programs.

\begin{example}
Let $ q = 13 $, $ n = 6 $, $ k = 3 $, $ h = 1 $, and $ \va = (1, 2, 3, 7, 8, 9) $. 
Put $ (a, b, \delta, \eta) = (2, 7, 2, 9) $ and $ \vx = (1, 8, 1, 5, 5, 1, 2, 7) $. 
By \cref{th:crdp}, both the non-GRS type $ [8,3] $ code $ \mathcal{C} $ and the $ [8,4] $ code generated by the matrix $ \big( \frac{G}{\vx} \big) $ are MDS over $ \f_{13} $, the covering radius of $ \mathcal{C} $ is $ \rho(\mathcal{C}) = 5 $ and $ \vx $ is a deep hole of $ \mathcal{C} $.
\end{example}

\begin{example}
Let $ q = 13 $, $ n = 6 $, $ k = 3 $, $ h = 0 $, and $ \va = (2, 3, 6, 8, 9, 10) $. 
Put $ (a, b, \delta, \eta) = (0, 1, 2, 8) $ and $ \vx = (8, 1, 8, 5, 1, 12, 0, 1) $. 
By \cref{cor:h=0}, both the non-GRS type $ [8,3] $ code $ \mathcal{C} $ and the $ [8,4] $ code generated by the matrix $ \big( \frac{G}{\vx} \big) $ are MDS over $ \f_{13} $, the covering radius of $ \mathcal{C} $ is $ \rho(\mathcal{C}) = 5 $ and $ \vx $ is a deep hole of $ \mathcal{C} $.
\end{example}

\begin{example}
Let $ q = 7 $, $ n = 5 $, $ k = 3 $, $ h = 1 $, and $ \va = (1, 2, 4, 5, 6) $. 
Put $ (a, b, \delta, \eta) = (6, 1, 3, 2) $ and $ \vx = (1, 1, 1, 6, 6, 6, 1) $. 
By \cref{th:crdp}, the $ [7,3] $ code $ \mathcal{C} $ is AMDS and the $ [7,4] $ code generated by the matrix $ \big( \frac{G}{\vx} \big) $ is MDS over $ \f_{7} $, the covering radius of $ \mathcal{C} $ is $ \rho(\mathcal{C}) = 4 $ and $ \vx $ is a deep hole of $ \mathcal{C} $.
\end{example}
\begin{example}
Let $\gamma$ be a primitive element of $ \f_8 $, $ n = 7 $, $ k = 5 $, $ h = 0 $, and $ \va = (1, \gamma, \gamma^3, \gamma^4, \gamma^5, \gamma^6, 0) $. 
Put $ (a, b, \delta, \eta) = (\gamma^3, \gamma^2, 1, \gamma^5) $ and $ \vx = (1, \gamma^5, \gamma, \gamma^6, \gamma^4, \gamma^2, 0, \gamma^3, \gamma^2) $. 
By \cref{th:crdp}, the $ [9,5] $ code $ \mathcal{C} $ is AMDS and the $ [9,6] $ code generated by the matrix $ \big( \frac{G}{\vx} \big) $ is MDS over $ \f_{8} $, the covering radius of $ \mathcal{C} $ is $ \rho(\mathcal{C}) = 4 $ and $ \vx $ is a deep hole of $ \mathcal{C} $.
\end{example}
 
\begin{remark}\label{thm: crdh-dual(C1)}
The covering radius and deep holes of $ C_1^{\perp} $ can also be characterized. 
In fact, similar to the proof of \cref{th: MDSh>0}, we know that code $ C_1 $ is MDS if and only if Conditions $ (1) $ and $ (2) $ in \cref{th: MDSh>0} are satisfied. 
\cref{th: MDSh>0} indicates when $ \mathcal{C} $ becomes MDS.
So, by \cref{lem: deephole}, if $ \mathcal{C} $ is MDS,
then the covering radius $ \rho(C_1^{\perp}) = k $ and the vector $ \ve $ in \cref{thm: C=barC_1} is a deep hole of $C_1^{\perp}$ for integers $ k, n, h $ satisfying $ 3 \leqslant k < k + 1 < n \leqslant q $ and $ 0 \leqslant h \leqslant k-2 $.   
\end{remark}


\section{Conclusion}\label{sec6}

In this paper, we constructed a class of extended TGRS codes and investigated some properties of these codes. The main results of this paper are as follows.
\begin{enumerate}
\item We showed that $ \mathcal{C} $ is an extended code of the type $ \overline{C}(\ve) $ in \eqref{eq: bar_C} in \cref{thm: C=barC_1}. 
\item We proved that $ C_1 $ is non-GRS type based on the dimension of the Schur product in \cref{thm: C_1noGRS}, and using this result, we further verified that the extended code $ \mathcal{C} $ of $ C_1 $ is also  non-GRS type in \cref{thm: CnonGRS}.
\item The sufficient and necessary conditions for $ \mathcal{C} $ being non-GRS type MDS were given in \cref{th: MDSh>0} and \cref{th: MDSh=0}.
\item The sufficient and necessary conditions for $ \mathcal{C} $ being an AMDS code were provided in \cref{th: AMDSh>0} and \cref{th: AMDSh=0}.
\item In \cref{lem: AMDS-CR}, for an $[n,k]$ AMDS code, we determined the sufficient and necessary conditions for its covering radius reaching $n-k$. As a consequence, the covering radii and deep holes of  $ \mathcal{C} $ were represented in \cref{th:crdp}, \cref{cor:a=b=0} and \cref{cor:h=0}, whenever $ \mathcal{C} $ is MDS or AMDS.
\end{enumerate}


\section*{Acknowledgements}

This work was supported  by the Natural Science Foundation of Shandong Province under Grants ZR2025MS65 and ZR2023MA042.



\end{document}